\documentclass[preprint, 12pt]{elsarticle}

\usepackage[T1]{fontenc}
\usepackage{mdframed}
\usepackage[most]{tcolorbox}
\usepackage{todonotes}
\usepackage{graphicx} 
\usepackage{hyperref}
\usepackage{amssymb,amsmath,amsthm}
\usepackage{amsmath,amssymb}
\usepackage{mathtools}
\usepackage{graphicx,xcolor}
\usepackage{amsmath}
\usepackage{xspace}
\usepackage[utf8]{inputenc}
\usepackage{amsthm}
\usepackage[normalem]{ulem} 
\usepackage{cleveref}
\usepackage{soul}

\newcommand{\mI}{\mathcal{I}}
\newcommand{\C}{\mathcal{C}}

\newcommand{\xp}{{\sf XP}\xspace}
\newcommand{\fpt}{{\sf FPT}\xspace}

\newcommand{\np}{\textup{\textsf{NP}}\xspace}

\newcommand{\nph}{{\textsf{NP}\textrm{-hard}\xspace}}
\newcommand{\npc}{{\textsf{NP}\textrm{-complete}\xspace}}
\newcommand{\woh}{\textsf{W[1]}\textrm{-hard}\xspace}

\newcommand{\wth}{\textsf{W[2]}\textrm{-hard}\xspace}

\newcommand{\true}{\textup{\textsc{true}}\xspace}
\newcommand{\false}{\textup{\textsc{false}}\xspace}

\newcommand{\lapxh}{\textrm{log-APX-hard}\xspace}

\newcommand{\defparproblem}[4]{
\begin{tcolorbox}[colback=gray!5!white,colframe=gray!75!black]
  \vspace{-1mm}
  \begin{tabular*}{\textwidth}{@{\extracolsep{\fill}}lr} #1  & {\bf{Parameter:}} #3 \\ \end{tabular*}
  {\bf{Input:}} #2  \\
  {\bf{Question:}} #4
  \vspace{-1mm}
\end{tcolorbox}
}

\newcommand{\mcs}{\textup{\textsc{MCS}}\xspace}
\newcommand{\mscs}{\textup{\textsc{MSCS}}\xspace}
\newcommand{\mcss}{\textup{\textsc{MCSS}}\xspace}
\newcommand{\mss}{\textup{\textsc{MSS}}\xspace}

\newcommand{\inn}{\textsc{\tiny in}\xspace}
\newcommand{\out}{\textsc{\tiny out}\xspace}
\newcommand{\sib}{\textsc{\tiny sib}\xspace}

\newcommand{\dinv}{\delta_v^{\textsc{{\tiny in}}}\xspace}
\newcommand{\doutv}{\delta_v^{\textsc{{\tiny out}}}\xspace}
\newcommand{\dsibv}{\delta_v^{\textsc{{\tiny sib}}}\xspace}
\newcommand{\cinv}{C_v^{\textsc{{\tiny in}}}\xspace}
\newcommand{\coutv}{C_v^{\textsc{{\tiny out}}}\xspace}
\newcommand{\csibv}{C_v^{\textsc{{\tiny sib}}}\xspace}

\newcommand{\dminv}{\delta_v^{\textsc{{\tiny min}}}\xspace}
\newcommand{\cminv}{C_v^{\textsc{{\tiny min}}}\xspace}

\newcommand{\dinvs}{\delta_S^{\textsc{{\tiny in}}}\xspace}
\newcommand{\doutvs}{\delta_S^{\textsc{{\tiny out}}}\xspace}
\newcommand{\dsibvs}{\delta_S^{\textsc{{\tiny sib}}}\xspace}
\newcommand{\cinvs}{C_S^{\textsc{{\tiny in}}}\xspace}
\newcommand{\coutvs}{C_S^{\textsc{{\tiny out}}}\xspace}
\newcommand{\csibvs}{C_S^{\textsc{{\tiny sib}}}\xspace}

\newcommand{\Eout}{\mathcal{E}_i^{^\textsc{{\tiny out}}}\xspace}
\newcommand{\Esib}{\mathcal{E}_i^{^\textsc{{\tiny sib}}}\xspace}

\newcommand{\din}{\delta^{^\textsc{{\tiny in}}}\xspace}

\newcommand{\svin}{S_v^\textsc{\tiny in}\xspace}
\newcommand{\svout}{S_v^\textsc{\tiny out}\xspace}
\newcommand{\svsib}{S_v^\textsc{\tiny sib}\xspace}

\newcommand{\NN}{\mathrm{NN}\xspace}

\newtcolorbox{coloredframe}[3][]{
    empty,
    breakable=true,
    sharp corners=all,
    top=4mm, left=4mm,
    borderline west={1.5pt}{0pt}{#3}, borderline north={1.5pt}{0pt}{#3},
    attach boxed title to top left={yshift=-1.75ex,xshift=6ex},
    coltitle=black,
    colback=white, colbacktitle=white,
    fonttitle=\bfseries,
    boxed title style={boxrule=0pt,colframe=white},
    title=#2,
    #1
}

\newtheorem{theorem}{Theorem}
\newtheorem{lemma}[theorem]{Lemma}
\newtheorem{definition}[theorem]{Definition}
\newtheorem{observation}[theorem]{Observation}

\begin{document}

\begin{frontmatter}



\title{New Complexity and Algorithmic Bounds for Minimum Consistent Subsets}


\author[label1]{Aritra Banik}

\author[label5]{Sayani Das}

\author[label2]{Anil Maheshwari}

\author[label3]{Bubai Manna}

\author[label4]{Subhas C Nandy}

\author[label1]{Krishna Priya K M}

\author[label3]{Bodhayan Roy}

\author[label4]{Sasanka Roy}

\author[label1]{Abhishek Sahu}

\affiliation[label1]{organization={National Institute of Science, Education and Research, An OCC of Homi Bhabha National Institute},
             addressline={Bhubaneswar},
             postcode={752050},
             state={Odisha},
             country={India}}

\affiliation[label5]{organization={Department of Mathematics, Mahindra University, Hyderabad, India},
             addressline={Hyderabad},
             country={India}}
             
\affiliation[label2]{organization={School of Computer Science, Carleton University},
             addressline={Ottawa},
             postcode={ON K1S 5B6},
             country={Canada}}

\affiliation[label3]{organization={Department of Mathematics, Indian Institute of Technology, Kharagpur},
             postcode={721302},
             state={West Bengal},
             country={India}}
             
\affiliation[label4]{organization={Advanced Computing and Microelectronics Unit, Indian Statistical Institute},
             addressline={Kolkata},
             state={West Bengal},
             postcode={700108},
             country={India}}


\begin{abstract}
In the Minimum Consistent Subset (\mcs) problem, we are presented with a connected simple undirected graph $G$, consisting of a vertex set $V(G)$ of size $n$ and an edge set $E(G)$. Each vertex in $V(G)$ is assigned to a color from the set $\{1,2,\ldots, c\}$. The objective is to determine a subset $S \subseteq V(G)$ with minimum possible cardinality, such that for every vertex $v \in V(G)$, at least one of its nearest neighbors in $S$ (measured in terms of the hop distance) shares the same color as $v$. A variant of \mcs is the minimum strict consistent subset (\mscs) in which instead of requiring at least one nearest neighbor of $v$, all the nearest neighbors of $v$ in $S$ must have the same color as $v$. The decision problem for \mcs, which asks whether there exists a subset $S$ of cardinality at most $l$ for some positive integer $l$, is known to be $\npc$ even for planar graphs.

{In this paper, we first show that the \mcs problem is \lapxh on general graphs. It is also $\npc$ on trees. 
We also provide a fixed-parameter tractable (\fpt) algorithm for \mcs on trees parameterized by the number of colors ($c$) running in $\mathcal{O}(2^{6c}n^6)$ time, significantly improving the currently known best algorithm whose running time is $\mathcal{O}(2^{4c}n^{2c+3})$.} In an effort to better understand the computational complexity of the \mcs problem across different graph classes, we extend our investigation to interval graphs. We show that it remains $\npc$ for interval graphs, thus adding to the family of graph classes where $\mcs$ remains intractable. For $\mscs$, we show that the problem is $\lapxh$ on general graphs and $\npc$ on planar graphs\footnote{A preliminary version of this article appeared in the Proceedings of the 44th Conference on Foundations of Software Technology and Theoretical Computer Science (FSTTCS 2024).}.
\end{abstract}

\end{frontmatter}



\section{Introduction}

For many supervised learning algorithms, the input comprises a colored training dataset $T$ in a metric space $(\mathcal{X},d)$ where each element $t\in T$ is assigned a color $C(t)$ from $[c]$. 
The objective is to preprocess $T$ in a manner that enables rapid assignment of a color to any uncolored element in $\mathcal{X}$, satisfying specific optimality criteria. 
One commonly used optimality criterion is the nearest neighbor rule, where each uncolored element $x$ is assigned a color based on the colors of its $k$ nearest neighbors in $T$ (where $k$ is a fixed integer). 
The efficiency of such an algorithm relies on the size of the training dataset. 
Therefore, it is crucial to reduce the size of the training dataset while preserving distance information. 
This concept was formalized by Hart~\cite{hart1968} in 1968 as the Minimum Consistent Subset (\mcs) problem. 
In this problem, given a colored training dataset $T$, the objective is to find a subset $S\subseteq T$ of minimum cardinality such that for every point $t\in T$, the color of $t$ is the same as the color of one of its nearest neighbors in $S$. 
Since its inception, the \mcs problem has found numerous applications, as evidenced by over $2800$ citations to \cite{hart1968} in Google Scholar.   

The \mcs problem for points in $\Re^2$ under the Euclidean norm is shown to be $\npc$ if the input points are colored with at least three colors. 
Furthermore, it remains $\npc$ even with two colors \cite{Wilfong, Khodamoradi}. 
Recently, it has been shown that the \mcs problem is \woh when parameterized by the output size \cite{Chitnis22}. 

In this paper, we explore the minimum consistent subset problem when $(\mathcal{X},d)$ is a graph metric. We use $[n]$ to denote the set of integers $\{1,\ldots, n\}$. For any graph $G$, we denote the set of vertices of $G$ by $V(G)$ and the set of edges by $E(G)$. Consider any graph $G$ and an arbitrary vertex coloring function $C:V(G)\rightarrow [c]$. For $U\subseteq V(G)$, let $C(U)={C(u):u\in U}$ denote the colors in $U$. For any two vertices $u,v\in V(G)$, the number of edges in the shortest path between $u$ and $v$ in $G$ is called the distance between $u$ and $v$, denoted by $d(u, v)$. This distance is also referred to as the \emph{hop-distance} between $u$ and $v$. For a vertex $v\in V(G)$ and any subset of vertices $U\subseteq V(G)$, let $d(v,U)=\min_{u\in U}d(v,u)$. The nearest neighbors of $v$ in the set $U$ are denoted as $\NN(v,U)$, formally defined as $\NN(v,U)=\{u\in U:d(v,u)=d(v,U)\}$.

For any graph $G$ and vertex $v \in V(G)$, let 
$N(v) = \{u \in V(G) : uv \in E(G)\}$ denote the (open) neighborhood of $v$, 
and let $N[v] = N(v) \cup \{v\}$ denote its closed neighborhood. We denote the distance between two subgraphs $G_1$ and $G_2$ in $G$ by $d(G_1,G_2)=\min \{d(v_1,v_2):v_1\in V(G_1),v_2\in V(G_2)\}$. For any subset of vertices $U\subseteq V(G)$ in a graph $G$, $G[U]$ denotes the subgraph of $G$ induced on $U$. Most of the symbols and notations of graph theory used are standard and taken from~\cite{diestel2012graph}. 

Suppose $G = (V, E)$ is a given connected and undirected graph, where vertices are partitioned into $c$ color classes, namely $V_1, V_2, \cdots, V_c$. This means that each vertex of $V$ has a color from the set of colors $\{1,2,\dots,c\}$, and each vertex in the set $V_i$ has color $i$. Therefore, $\cup_{i=1}^{c}V_i=V$, and $V_i\cap V_j=\phi$ for $i\neq j$. Throughout the paper, by a $c$-colored graph, we mean that each vertex of the graph has been assigned one of $c$ colors, and adjacent vertices are allowed to have the same color. For any set $S$, we denote its cardinality by $\lvert S\rvert$. Next, we formally define the minimum consistent subset and minimum strict consistent subset problems.

\begin{definition} {\bf Minimum Consistent Subset(\mcs) Problem}\\
 A \emph{Minimum Consistent Subset (\mcs)} is a subset $S \subseteq V$ of minimum cardinality such that for every vertex $v \in V$, if $v \in V_i$ then $\NN(v,S)\cap V_i \ne \emptyset$. 
\end{definition}

\begin{definition}{\bf Minimum Strict Consistent Subset(\mscs) Problem}\\
 A subset $S\subseteq V$ is said to be an \mscs if, for each vertex $v\in V$, every vertex in the set of nearest neighbors of $v$ in $S$ (that is $\NN(v,S)$) has the same color as $v$ and $\lvert S \rvert$ is minimum. 
\end{definition}

\begin{figure}[h]
\centering
\includegraphics[width=0.9\textwidth]{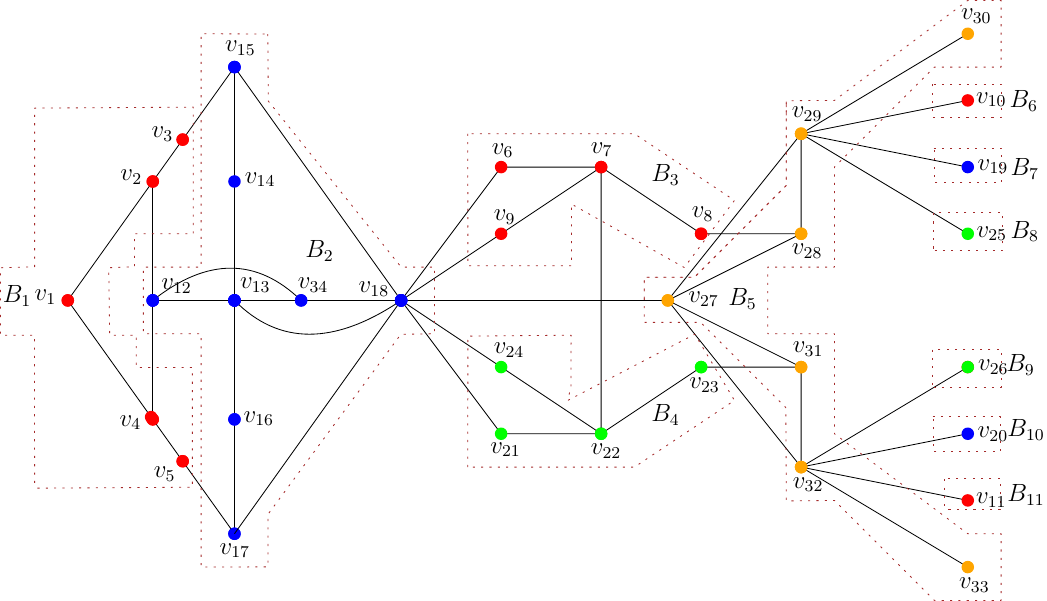}
\caption{$C=\{red,\text{ }blue,\text{ } $ $ green,\text{ } orange\}$ represent set of colors and the corresponding classes are $V_{red}=\{v_1,$ $\dots, $ $ v_{11}\}$, $V_{blue}=$ $\{v_{12},$ $\dots,$ $v_{20}, v_{34}\}$, $V_{green}=$ $\{v_{21},$ $\dots,$ $v_{26}\}$, and $V_{orange}=$ $\{v_{27},$ $\dots,$ $v_{33}\}$. The sets of vertices $\{v_{30}, v_{10}, v_{19}, v_{25}\}$ and $\{v_1, v_{13}, v_7, v_{22}, v_{29}, v_{32}, v_{10}, v_{19}, v_{25}, v_{26}, v_{20}, v_{11}\}$ forms \mcs and \mscs, respectively.
Also, $\{v_{26},  v_{20}, v_{11}, v_{33}\}$ and $\{v_1,v_{34}, v_7, v_{22}, v_{29}, v_{32}, v_{10}, v_{19}, v_{25}, v_{26}, v_{20}, v_{11}\}$,
are \mcs and \mscs, respectively.  The vertices inside the brown-dotted region form a block, and the blocks are labeled as $B_1, \dots, B_{11}$.} \label{ex1}
\end{figure}

The examples of \mcs and \mscs are shown in \Cref{ex1}. There may be more than one \mcs in an example; however, its size remains unchanged for the same example. Similarly, there may be multiple \mscs in a given example; however, their sizes remain unchanged for the same example.

The decision version of the \mcs and \mscs problems on graphs is defined as follows:

\begin{tcolorbox}[enhanced,title={\color{black} 
\sc{Decision version of \mcs and \mscs problems}}
colback=white, boxrule=0.4pt, attach boxed title to top center={xshift=-0.8 cm, yshift*=-2mm}, boxed title style={size=small,frame hidden,colback=white}]
		
    \textbf{Input:} A graph $G$, a coloring function $C:V(G)\rightarrow [c]$, and an integer $\ell$.\\
    \textbf{Question 1:} Does there exist a consistent subset of size $\le \ell$ for $(G,C)$? 
    
    \textbf{Question 2:} Does there exist a strict consistent subset of size $\le \ell$ for $(G,C)$? 

\end{tcolorbox}

Banerjee et al.~\cite{BBC} proved that the \mcs problem is \wth \cite{book} when parameterized by the size of the minimum consistent set, even when limited to two colors, thus rulling the possibility of an \fpt algorithm parameterized by $(c+\ell)$ under standard complexity-theoretic assumptions for general graphs. This naturally raises the question of determining the simplest graph classes where the problem remains computationally intractable. Dey et al.~\cite{DeyMN23} presented a polynomial-time algorithm for \mcs on simple graph classes such as paths, spiders, combs, and caterpillars. The \mcs has gained significant research attention in recent years, particularly when the underlying graph is a tree.  Dey et al.~\cite{DeyMN21} presented a polynomial-time algorithm for bi-colored trees, and Arimura et al.~\cite{Arimura23} presented an \xp algorithm parameterized by the number of colors $c$, with a running time of $\mathcal{O}(2^{4c}n^{2c+3})$.

Biniaz and Khamsepour \cite{biniaz2024minimum} presented a polynomial-time algorithm for the \emph{minimum consistent spanning subset} (\mcss) in trees. The minimum consistent spanning subset problem is a variant of \mcs and is defined as finding a subset $S\subseteq V(G)$, of minimum cardinality, such that for every vertex $v\in V(G)$, if $v\in V_i$, then $\NN (v,S)\cap V_i\neq \emptyset$ and for each block $B_i$, $B_i\cap S\neq \emptyset$. In \Cref{bubailemma101}, we show that each block must contain at least one vertex of every strict consistent subset. Therefore, the algorithm for \mscs in trees \cite{manna2024minimumstrictconsistentsubset} is quite similar to the algorithm for \mcss. However, minor adjustments in graph settings may be required to find \mscs in trees. For example, in \Cref{ex1}, if we select $v_3$ in \mcss, then we must also select either $v_{15}$ or at least one of $v_{14}$ and $v_{18}$. However, if we take $v_3$ in \mscs, then we must select $v_{15}$ in \mscs. Thus, in \mscs, choosing a vertex in one block may constrain choices in another, unlike in \mcss. Nevertheless, the computational time and overall algorithmic logic remain the same. Additionally, we observe that neighbor blocks must depend on each other when finding a solution for \mscs and \mcss. 

Wilfong \cite{Wilfong} defined two problems, \mcs and the \emph{Minimum Selective Subset} (\mss). For graph settings, \(\mss\) is studied in \cite{BBC}, where it is proved that the \(\mss\) problem is \(\npc\) on general graphs. Further algorithms and hardness reductions for trees and planar graphs are given in \cite{manna2025minimumselectivesubsetgraph}.

{\bf New Results:}
First, we show that the \mcs problem is \lapxh on general graphs in Section \ref{lapxhmcs}. The most intriguing question yet to be answered is whether $\mcs$ remains $\nph$  for trees \cite{Arimura23, DeyMN21}. In this paper, we decisively answer this question in the affirmative. This is particularly noteworthy given the scarcity of naturally occurring problems known to be $\nph$ on trees. Our contribution includes a reduction from the MAX-2SAT problem, detailed in \Cref{hardnesstree}. Next, we show that \mcs is fixed-parameter tractable (\fpt) for trees on $n$ vertices, significantly improving the results presented in 
Arimura et al.~\cite{Arimura23}. Our intricate dynamic programming algorithm runs in  $\mathcal{O}(2^{6c}{n}^6)$ time, whereas 
\cite{Arimura23} requires $\mathcal{O}(2^{4c}n^{2c+3})$ time; see \Cref{fpt}. 

Moreover, in  \Cref{hardnessint}, we show that \mcs on interval graphs is \nph. While interval graphs are unrelated to trees, our hardness result for interval graphs raises new questions about the fixed-parameter tractability of \mcs on this graph class. 

For $\mscs$, we show that the problem is $\lapxh$ on general graphs and $\npc$ on planar graphs in  \Cref{section7} and \Cref{npmscs}, respectively.

\section{Preliminary Results}\label{pre}

In this section, we state some definitions and preliminary results. 

\begin{observation}
If all vertices of a graph $G$ have the same color, i.e., $G$ is monochromatic, then every vertex of $G$ is both an \mcs and an \mscs.
\end{observation}
In the rest of the paper, we assume that $G$ is not monochromatic. Otherwise, the consistent set problems have trivial solutions. 

\begin{observation}\label{observation6}
Every strict consistent subset is a consistent subset.
\end{observation}

\begin{definition} {\bf Block}\\
 For any graph $G$, a \emph{block} is a maximal connected set of vertices sharing the same color. 
\end{definition}
In \Cref{ex1}, the vertices inside the brown-dotted region form a block.
\begin{lemma}{{\label{bubailemma101}}}
For any graph $G$ and an arbitrary set of colors $C$, each block must contain at least one vertex in every strict consistent subset.
\end{lemma}

\begin{proof}
Let $S$ be a strict consistent subset. Suppose $B$ is a block and $S$ does not contain any vertex from $B$. Then every vertex of $B$ must have its nearest vertex in $S$ from a different block of the same color, say $B'$. Assume $v'$ (which is in $B'$) is included in $S$, and $v$ (which is in $B$) has $v'$ as its nearest neighbor in $S$. Since $B$ and $B'$ are distinct blocks, the shortest path between $v$ and $v'$ must contain at least one vertex of a different color from $v$. We assume that $v''$ is such a vertex with a different color than $C(v)$. Now we have two cases:
\begin{itemize}
    \item If $v''\in S$, then $v$ would have $v''$ as its nearest neighbor instead of $v'$ with $C(v)\neq C(v'')$, contradicting the strict consistency of $S$. Thus, $v''\notin S$.
    \item Since $v''\notin S$, it must have a nearest neighbor $x\in S$ of the same color such that $d(v'',x)<d(v'',v')$; otherwise, at least one nearest neighbor of $v''$ in $S$ would have different color, contradicting the definition of a strict consistent subset. Now $d(v,x)\leq d(v,v'')+d(v'',x)< d(v,v'')+d(v'',v')$, which implies $v$ has $x$ as its nearest neighbor in $S$ with $C(v)\neq C(x)$, leading to a contradiction.
\end{itemize}
Thus the lemma holds for strict consistent subsets. 

\end{proof}

\section{log-APX Approximation of MCS on General Graphs}\label{lapxhmcs}
We prove that \mcs is \lapxh (see \cite{10.5555/1965254} for definitions of various complexity classes). We reduce the Minimum Dominating Set problem to the Consistent Subset problem. In the {\em Dominating Set} problem, given a graph $G$ and an integer $k$, the objective is to decide whether there exists a subset $U\subseteq V(G)$ of size $k$ such that for any vertex $v\in V(G)$, $N[v]\cap U\neq \emptyset$. It is known that the minimum Set Cover problem is \lapxh, i.e., it is $\nph$ to approximate within a $c\cdot \log n$ factor. factor~\cite{raz97}. As there exists an $L$-reduction from Set Cover to the Dominating Set problem, the latter is \lapxh.

\begin{figure}[t]
\centering
\includegraphics[scale=1.1]{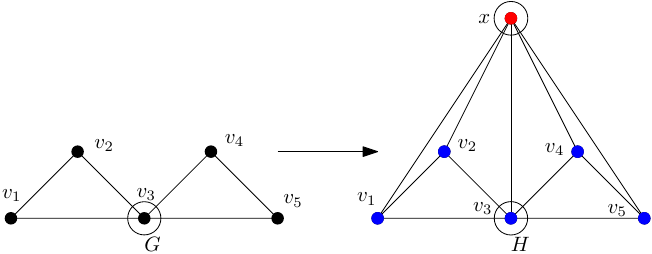}
\caption{Example showing the reduction of a Dominating Set instance of graph $G$ into an equivalent \mcs instance of graph $H$. Vertices inside small circles indicate the solution.}
\label{fig:bubai}
\end{figure}

\textbf{Reduction.} Let $G=(V(G),E(G))$ be a graph. We construct an instance of the consistent subset problem $H=(V(H), E(H))$ as follows.
\begin{enumerate}
    \item $V(H)=V(G)\cup\{x\}$.
    \item $E(H)=E(G)\cup \{(x,v):v\in V(G)\}$.
    \item For all $v \in V(H)\setminus \{x\}$, we set $C(v)=1$. We also set $C(x)=2$.
\end{enumerate}

For the sake of completeness, we state the following lemma.

\begin{lemma}\label{bubailemma6}
There exists a Dominating Set for $G$ of size at most $k$ if and only if there is a consistent subset of size at most $k+1$ for the graph $H$.
\end{lemma}
 \begin{proof}
    Let $D$ be a Dominating Set of size $k$ for the graph $G$. We claim that $D'=\{x\}\cup D$ is a consistent subset of $H$. If not, then there is a vertex $v_i\in V(H)\setminus D'$ such that $d(v_i,D)>d(v_i,x)=1$. This contradicts the assumption that $D$ is a Dominating Set, and hence the claim holds.

    On the other hand, suppose $D'$ is a consistent subset of size $k+1$ in the graph $H$. Observe that $x\in D'$ as $x$ is the only vertex with color $2$. We claim that $D=D'\setminus \{x\}$ is a Dominating Set of $G$. If not, then there is a vertex $v\in V(G)\backslash D$ such that $N(v)\cap D=\emptyset$. Thus $d(v,D')>1$ but $d(v,x)=1$ and $C(v)\neq C(x)$. This contradicts the assumption that $C$ is a consistent subset and hence the claim holds.
\end{proof}
From \Cref{bubailemma6}, we have the following theorem.
\begin{theorem}
    There exists a constant $c>0$ such that it is $\nph$ to approximate the \mcs problem within a factor of $c\cdot \log n$.
\end{theorem}

\section{NP-hardness of MCS on Trees} \label{hardnesstree}

In this section, we prove that \mcs is $\npc$ when the input graph is a tree and the number of colors is arbitrary (that is $c$ is not constant). We present a reduction from the MAX-2SAT problem to \mcs. Let $\theta$ be a given MAX-$2$SAT formula with $n$ variables $\{x_1, \ldots, x_n\}$ and $m$ clauses $\{c_1, \ldots, c_m\}$, $n,m \geq 50$. We construct an instance $(T_{\theta}, C_{\theta})$ of the MCS problem from $\theta$ as follows.

\begin{tcolorbox}[breakable,bicolor,
  colback=cyan!5,colframe=cyan!5,title=Interval Graph Construction.,boxrule=0pt,frame hidden]

\noindent \underline{\textbf{\large Construction of $(T_{\theta}, C_{\theta})$.}}

\smallskip

\noindent The constructed tree $T_{\theta}$ is composed of variable gadgets, clause gadgets, and central vertex gadgets. 
\smallskip

\noindent {\bf Variable Gadget.}

\noindent A variable gadget $X_i$ for the variable $x_i \in \theta$ has two components where each component has a literal path and $M$ pairs of stabilizer vertices, as described below (see \Cref{fig:red}), where $M$ is very large (we will define the exact value of $M$ later).

\smallskip

\noindent {\bf Literal paths:} The two literal paths of the variable gadget $X_i$ are $P^\ell_i=\langle x_i^1, x_i^2,x_i^3,x_i^4\rangle$ and $\overline{P}^\ell_i = \langle \overline{x}_i^1,\overline{x}_i^2,\overline{x}_i^3, \overline{x}_i^4\rangle$, each consisting of four vertices; they are referred to as \emph{positive literal path} and \emph{negative literal path}, respectively. Here, by a path of $k$ ($> 2$) vertices, we mean a connected graph with $k-2$ vertices of degree $2$ and the remaining two nodes having degree 1. All the vertices on the path $P_i^\ell$ are of color $c^\ell_i$ and all the vertices on the path $\overline{P}_i^\ell$ are of color $\overline{c}^\ell_i$.
 
\smallskip

\noindent {\bf Stabilizer vertices:} $M$ pairs of vertices $\{s_{i}^1, \overline{s}_i^1\}, \ldots, \{s_{i}^M, \overline{s}_i^M\}$, where the color of each pair of vertices $\{s_{i}^j, \overline{s}_i^j\}$ is $c^s(i,j)$. We denote the set of vertices $S_i=\{ s_{i}^1,  \ldots, s_{i}^M\}$ as  \emph{positive stabilizer vertices} and the set of vertices $\overline{S}_i=\{\overline{s}_i^1, \ldots,  \overline{s}_i^M\}$ as \emph{negative stabilizer vertices}. Each vertex in $S_i$ is connected to $x_i^1$ and each vertex in $\overline{S}_i$ is connected to $\overline{x}_i^1$.

The intuition behind this set of stabilizer vertices is that by setting a large value of $M$ we ensure that either $\{s_{i}^1,  \ldots, s_{i}^M\}$ or $\{\overline{s}_i^1, \ldots,  \overline{s}_i^M\}$ must be present in any ``small sized solution''.

\smallskip

\noindent {\bf Clause Gadget.}

\noindent For each clause $c_i=(y_i\lor z_i)$, where $y_i$ and $z_i$ are two (positive/negative) literals, we define the clause gadget $C_i$ as follows. It consists of three paths, namely \emph{left occurrence path} $P_i^Y=\langle y_i^1,\dots, y_i^7\rangle$, \emph{right occurrence path} $P_i^Z=\langle z_i^1,\dots, z_i^7\rangle$, and {\em clause path} $P_i^W=\langle w_i^1,\dots, w_i^7\rangle$ (see  \Cref{fig:red}).
All the vertices in $P_i^Y$ (resp. $P_i^Z$) have the same color as the corresponding literal path in their respective variable gadgets, i.e. for any literal, say $y_i$ in $C_i$, if $y_i=x_i$ (resp. $\overline{x_i}$) then we set the color of the vertices of $P_i^Y$ as $c_i^x$ (resp. $\overline{c}_i^x$). The color of the vertices on the path $P_i^Z$ is set in the same manner. We color the vertices in $P_i^W$ by $c^w_i$, which is different from that of the vertices in $P_i^Y$ and $P_i^Z$.
We create the clause gadget $C_i$ by connecting $y_i^1$ with $w_i^2$ and $z_i^1$ with $w_i^6$ by an edge (see  \Cref{fig:red}).

\smallskip

\noindent {\bf Central Vertex Gadget.}

\noindent We also have a central path $P^v=\langle v_1,v_2, v_3\rangle$. The color of all the vertices in $P^v$ is the same, say $c^v$, which is different from the color of all other vertices in the construction. For each variable gadget $X_i$,  $x_i^1$ and $\overline{x}_i^1$ are connected to the vertex $v_1$ (see  \Cref{fig:red}). For each clause $C_i$, $w_i^4$ is connected with $v_1$. The color of the vertices of $P^v$  is $c^v$.

\end{tcolorbox}

\begin{figure}[h]
\centering
\includegraphics[width=0.8\textwidth]{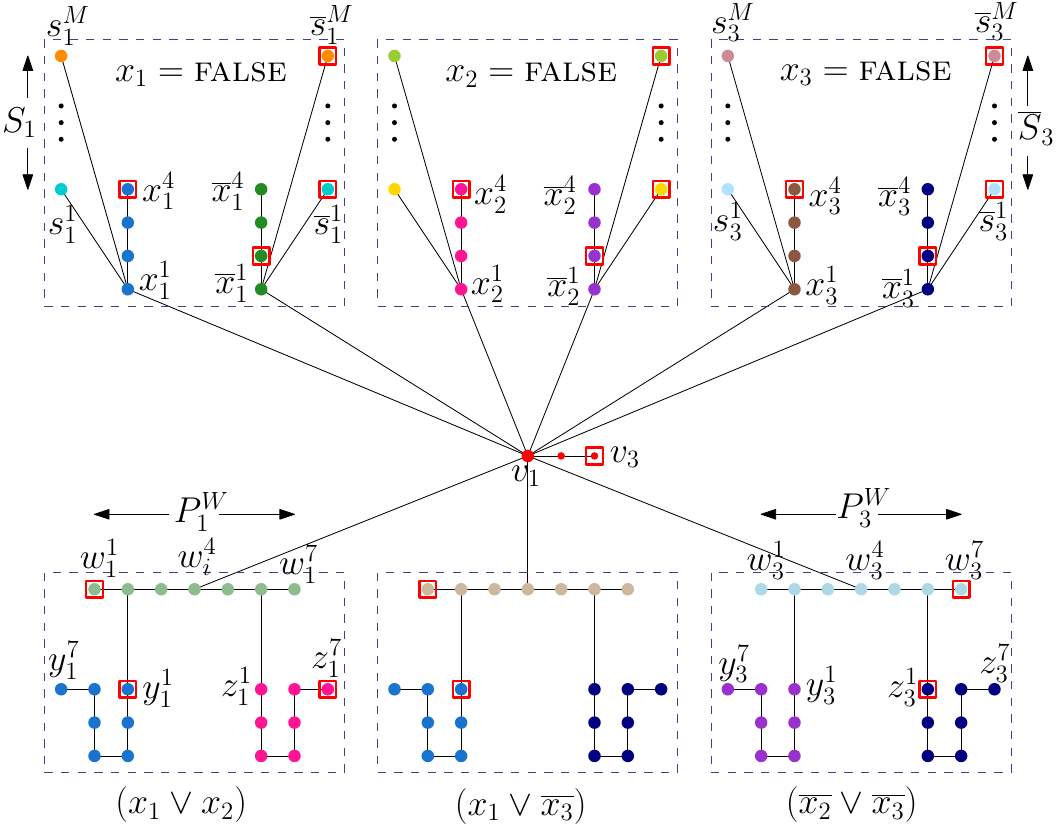}
\caption{An example of the construction of $(T_{\theta}, C_{\theta})$ is shown. For the assignment $x_1=x_2=x_3=\false$, the corresponding consistent subset is indicated with a red box around the vertices. In this assignment, $(x_1 \lor x_2)$ is not satisfied, whereas the rest of the clauses are satisfied.} \label{fig:red}
\end{figure}

Our objective is to show that there exists an \mcs of size at most $N(k)=n(M+2)+2k+3(m-k)+1$ in the tree $T_\theta$ if at least $k$ clauses of $\theta$ are satisfied; otherwise, the size is strictly greater than $N(k)$. We now prove a set of auxiliary claims about a minimum consistent subset for $(T_{\theta}, C_{\theta})$. 

\Cref{lemma:exactlyone}  states that by strategically choosing the vertices in a variable gadget, the vertices of the tree corresponding to that variable gadget can be consistently covered by choosing its only one set of stabilizer vertices.

\begin{lemma} \label{lemma:exactlyone}
For any consistent subset $V_C$ of size at most $N(k)=n(M+2)+2k+3(m-k)+1$ in the tree $T_\theta$, the following are true.
\begin{itemize}
    \item For any variable $x_i$, exactly one of the following is true.
        \begin{itemize}
            \item $S_i\subset V_C$, $\overline{S}_i\cap V_C=\emptyset$, and $x_i^2,\overline{x}_i^4\in V_C$. 
            \item $\overline{S}_i\subset V_C$, ${S}_i\cap V_C=\emptyset$, and $\overline{x}_i^2,x_i^4\in V_C$, 
        \end{itemize} and
    \item $v_3\in V_C$.
\end{itemize}
\end{lemma}

    \begin{proof}
     For a variable $x_i$, let $S_i\cap V_C \neq \emptyset$. Let $s_i^j\in S_i\cap V_C$. But then every vertex $v \in S_i\setminus \{s_i^j\}$ must have a vertex within distance 2 of its own color, since $d(s_i^j,v)=2$. Hence $S_i\subseteq V_C$. One can similarly prove that $\overline{S}_i\cap V_C \neq \emptyset
      \implies \overline{S}_i\subseteq V_C$.
     Also, every variable gadget contains $M$ uniquely colored vertices and hence has at least $M$ vertices in $V_C$. So, if $M>>n$, we have that $N(k)<(n+1)M$, and there exists no variable gadget that contains vertices from both $S_i$ and $\overline{S}_i$. In other words, exactly one of the following holds for every variable gadget corresponding to a variable $x_i$: 
     \begin{itemize}
     \item  $S_i\subset V_C$, $\overline{S}_i\cap V_C=\emptyset$
     \item $\overline{S}_i\subset V_C$, ${S}_i\cap V_C=\emptyset$
     \end{itemize}

     Below, we look into one of these cases, and a similar argument may be made for the other case.

     \noindent\textbf{\underline{Case 1: $S_i\subset V_C$, $\overline{S}_i\cap V_C=\emptyset$}}

     Notice that there must be a vertex in the literal path \{$x_i^1,x_i^2,x_i^3,x_i^4$\} of the variable gadget $X_i$ since $d(S_i,x_i^1)=1$ and the distance to any other vertex of the same color (other than these two vertices) is more than 1. But $x_i^1 \notin V_C$, as $\overline{S}_i\cap V_C=\emptyset$ and $d(x_i^1,\overline{S}_i)< d(S_i,\overline{S}_i)$. Hence $x_i^2 \in V_C$.

     Similarly there must be a vertex in the literal path $\{\overline{x}_i^1,\overline{x}_i^2,\overline{x}_i^3,\overline{x}_i^4\}$ of the variable gadget $X_i$ since $d(S_i,\overline{x}_i^1)=3$ and the distance to any other vertex of the same color (other than $\{\overline{x}_i^2,\overline{x}_i^3,\overline{x}_i^4\}$) is more than 3. And the distance of 4 between $S_i$ and $\overline{S}_i$ eliminates the possibility of any of $\overline{x}_i^1,\overline{x}_i^2$ or $\overline{x}_i^3$ belonging in $V_C$. Thus, $\overline{x}_i^4\in V_C$. 

     \noindent\textbf{\underline{Case 2: $\overline{S}_i\subset V_C$, ${S}_i\cap V_C=\emptyset$}} 
     
     Case 2 may be argued in a manner similar to Case 1.

   Moreover, $V_C$ must contain at least one vertex from the set $\{v_1, v_2, v_3\}$. However, the distance of 4 between $S_i$ and $\overline{S_i}$ rules out the possibility of either $v_1$ or $v_2$ being in $V_C$. Consequently, $v_3 \in V_C$.           
\end{proof}
   
To satisfy the inequality in the above lemma, we now set the value of $M$ as $n^3$.
In the next lemma, we present a bound on the vertices from each clause gadget that are contained in a consistent subset of size at most $N(k)$.
For any clause $C_i$, denote the corresponding clause gadget by $T_i^C=G[\{w_i^a,y_i^a,z_i^a|1\leq a\leq 7\}]$.
    \begin{lemma}
        In any consistent subset $V_C$ of the tree $T_\theta$, for each clause $C_i$, $2\leq |V(T_i^C)\cap V_C|$.\label{lemma:clause2or3}
    \end{lemma}
    \begin{proof}
There needs to be a vertex among the vertices $\{w_i^a \mid 1 \leq a \leq 7\}$ since they are distinctly colored from all other vertices. If this vertex belongs to $\{w_i^a \mid 1 \leq a \leq 4\}$, then there must also be a vertex in $\{y_i^a \mid 1 \leq a \leq 7\}$ since the nearest vertex of the same color (any $y_i^a$) is farther away than the vertex with the color of any $w_i^a$. Similarly, if this vertex belongs to $\{w_i^a \mid 4 \leq a \leq 7\}$, then there must be a vertex in $\{z_i^a \mid 1 \leq a \leq 7\}$. Therefore, $2 \leq |V(T_i^C)\cap V_C|$.
    \end{proof}

\begin{theorem}
    There exists a truth assignment of the variables in $\theta$ which satisfies at least $k$ clauses if and only if there exists a consistent subset of size at most $N(k)$ for $(T_{\theta},C_{\theta})$.
\end{theorem}

\begin{proof}
$(\Rightarrow)$ For the forward direction, let there exist an assignment $A$ to the variables of $\theta$ that satisfies $k$ clauses. Consider the following set of vertices $V_A$. For each variable $x_i$ if $x_i$ is $\true$, include all the vertices of $S_i$ in $V_A$. Also include $x_i^2$ and $\overline{x}_i^4$ in $V_A$. If $x_i=\false$ then include all the vertices of $\overline{S}_i$ in $V_A$. Also include $x_i^4$ and $\overline{x}_i^2$ in $V_A$.

For every satisfied clause $C_i=(y_i\lor z_i)$ with respect to $A$ we include the following vertices in $V_A$. Without loss of generality assume that $y_i=\true$. We include $w_i^7$ and $z_i^1$ in $V_A$.
For every unsatisfied clause $C_i=(y_i\lor z_i)$, we include  $w_i^1$, $y_i^1$ and  $z_i^7$ in $V_A$. We also include $v_3$ in $V_A$.

Observe that the cardinality of $V_A$ is $N(k)=n(M+2)+2k+3(m-k)+1$. Next, we prove that $V_A$ is a consistent subset for $T_{\theta}$. Observe that for any pair of vertices $(s_i^j,\overline{s}_i^j)$, exactly one of them is in $V_A$. Without loss of generality assume that $s_i^j\in V_A$. Observe that $d(s_i^j,\overline{s}_i^j)=d(\overline{s}_i^j,V_A)=4$. If $x_i=\true$ then $d(x_i^j,x_i^2)\leq d(x_i^j,V_A)$ and $d(\overline{x}_i^j,\overline{x}_i^4)\leq d(\overline{x}_i^j,V_A)$. The case where $x_i=\false$ is symmetric.

For any clause gadget either $w_i^1$ or $w_i^7$ is in $V_A$. Without loss of generality assume that $w_i^1\in V_A$. Observe that for every vertex $w_i^j$, $d(w_i^j,w_i^1)=d(w_i^j,V_A)$. Let $C_i=(y_i\lor z_i)$ be a satisfied clause and without loss of generality assume that $y_i=x_j=\true$. Observe that $d(y_i^1,x_j^2)=6$, and $d(y_i^3,V_A)=6$, and $d(y_i^a,x_j^2)=d(y_i^a,V_A)$.
For any unsatisfied clause $C_i=(y_i\lor z_i)$ observe that $d(y_i^a,x_j^2)=d(y_i^a,V_A)$. Also for any $v_j$ where $1\leq j\leq 3$, $d(v_j,v_3)=d(v_j,V_A)$. Therefore $V_A$ is a consistent subset for $T_{\theta}$. 
    
    \smallskip
\noindent $ (\Leftarrow) $
In the backward direction, let there be a consistent subset $V_C$ of size at most $N(k)$ for $(T_{\theta},C_{\theta})$. We know from \Cref{lemma:exactlyone} that either $S_i\subset V_C$ or $\overline{S}_i\subset V_C$. From \Cref{lemma:exactlyone}, any such solution has at least $n(M+2)+1$ vertices from $V_C$ outside the clause gadgets, leaving at most $2k+3(m-k)$ that may be chosen from the clause gadgets.

Each clause gadget comprises of vertices of three distinct colors: one color exclusive to the clause itself and two colors dedicated to literals. An essential insight is that if there are no vertices in $V_C$ of colors specific to the literals from a clause in the variable gadgets, then such a clause gadget must contain at least three vertices from $V_C$. This assertion is valid because the distance between two sets of vertices of the same color (corresponding to the same literal in two clauses) across any two clauses is at least 8, while vertices in $V_C$ of clause-specific colors are at a distance of at most 6.

This fact, coupled with \Cref{lemma:clause2or3}, implies that there are at least $k$ clauses for whom colors specific to at least one of their literals have the vertices in $V_C$ of the same color from the variable gadgets. Making the same literals true and setting other variables arbitrarily gives us an assignment that satisfies at least $k$ clauses.  
\end{proof}

\section{\mcs on Trees: A Parameterized Algorithm} \label{fpt}

This section considers the optimization version of the MCS problem for the trees. We provide a parameterized algorithm for the computation of MCS for trees.

\begin{definition}\textbf{Fixed-Parameter Tractable Time}\\ 
A problem is said to be solvable in \emph{fixed-parameter tractable (FPT) time} with respect to a parameter $k$ if it can be solved in time $f(k) \cdot n^{\mathcal{O}(1)}$, 
where $f$ is a computable function depending only on $k$, and $n$ is the size of the input \cite{book}. 
\end{definition}

In our setting, the parameter is $c$, the number of color classes, and our algorithm runs in time $\mathcal{O}(2^{6c}n^6)$, which is of the form $f(c) \cdot n^{\mathcal{O}(1)}$ and hence is FPT.

\defparproblem{\textsc{Minimum Consistent Subset for Trees} }
{A rooted tree $T$, whose vertices $V(T)$ are colored with a set $C$ of $c$ colors.}
{$c$}
{Find the minimum possible size of a consistent subset (\mcs) for $T$?}

We consider $T$ a rooted tree by taking an arbitrary vertex $r$ as its root. We use $V(T')$ to denote the vertices of a subtree $T'$ of $T$, and $C(U) \subseteq C$ to denote the subset of colors assigned to the subset of vertices $U  \subseteq V$, and $C(u)$ to denote the color attached to the vertex $u \in V(T)$. For any vertex $v$, let $\eta_v$ denote the number of children of $v$  and we denote the children of $v$ by $v_1,v_2,\cdots, v_{\eta_v}$. We denote the subtree rooted at a vertex $v$ by $T(v)$. For any vertex $v$ and any integer $i<\eta_v$, we use $T_{i+}(v)$ to denote the union of subtrees rooted at $v_{i+1}$ to $v_{\eta_v}$, and $T_i(v)$ to denote the subtree rooted at $v$ and containing first $i$ many children of $v$. Thus, $T_{i+}(v)=\cup_{i+1\leq j\leq \eta_v}T(v_j)$, which is a forest, and $T_i(v)=T(v)\setminus T_{i+}(v)$. In \Cref{fig:dpp}(a), the light yellow part is $T_i(v)$, and the light sky-colored part is $T_{i+}(v)$. We define $T^{out}(v)=T\setminus T(v)$. See \Cref{fig:dpp}a.

\begin{figure}[!htb]
    \centering
    \begin{minipage}{.4\textwidth}
        \centering
        \includegraphics[scale=.8]{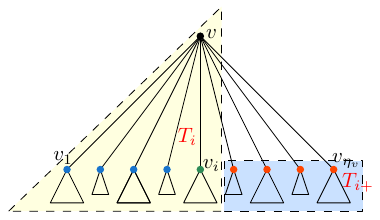}\\
            {\small (a)}\\
    \end{minipage}%
    \begin{minipage}{0.03\textwidth}

    \end{minipage}
    \begin{minipage}{0.57\textwidth}
        \centering
        \includegraphics[scale=0.4]{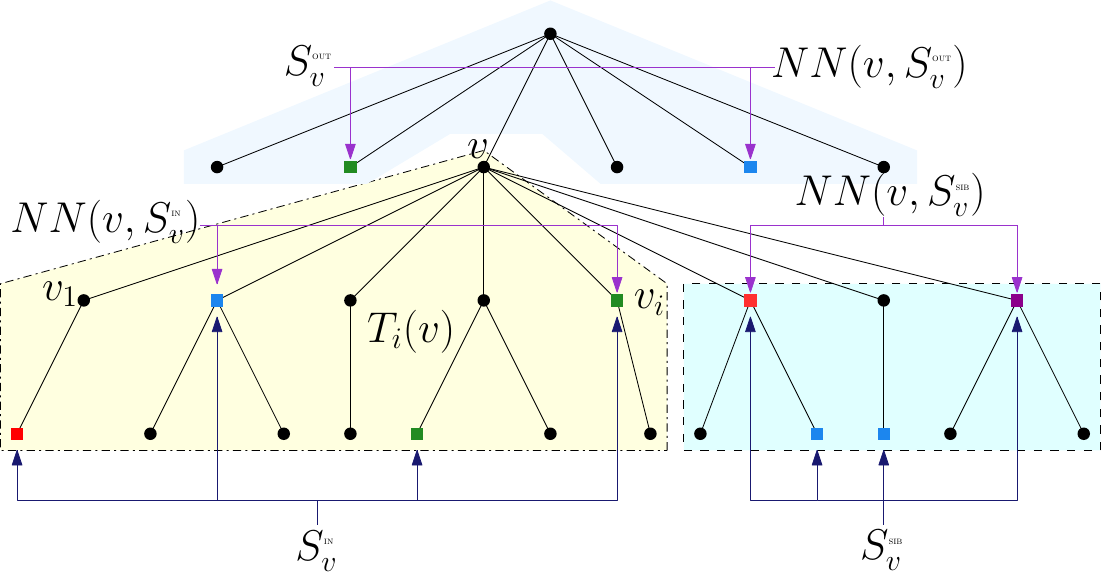}\\
            {\small (b)}\\
    \end{minipage}
    \caption{Illustration of the bottom-up dynamic programming routine. $\blacksquare$ vertices denote the consistent subset. $\dinvs = 1$, $\doutvs = 2$, and $\dsibvs = 1$. (a) Example of subtrees $T_i(v)$ and $T_{i+}(v)$ for each $v \in V(T)$. (b) Examples of $\svin$, $\svsib$, and $\svout$, along with their corresponding sets of nearest neighbors from vertex $v$: $\NN(v,\svin)$, $\NN(v,\svsib)$, and $\NN(v,\svout)$, respectively.}

  \label{fig:dpp}
\end{figure}

For any positive integer $d$ and for any vertex $v\in V(T)$, a set of vertices $U\subset V(T)$ is called $d$-equidistant from $v$ if $d(u_i,v)=d$ for all $u_i\in U$. Any subset of vertices $U$ spans a set of colors $C' \subseteq C$ if $C(U)=C'$. 

For any vertex $v\in V(T)$, we use $\Esib(v,d,C')$ (resp. $\Eout(v,d,C')$) to denote the set of subsets of vertices in $T_{i+}(v)$ (resp. $T\setminus T(v)$), which are $d$-equidistant from $v$ and span the colors in $C'$. 
Next, we define a ({\em partial}) {\em consistent subset} for a subtree $T_i(v)$.

\noindent{\bf Intuition:}
Our dynamic programming (DP) routine exploits the key observation that a (partial) consistent subset (formally defined below) for a subtree $T(v)$ can be computed in FPT time. This computation is possible given the distance to the closest vertex in the consistent subset that lies outside $T(v)$ and the colors of those vertices. The entries in our DP table store the minimum size of partial consistent subsets for all subtrees $T_i(v)$. These subsets are defined based on six parameters: distance to the closest vertex from $v$ in the consistent subset in $T_i(v)$, $T^{out}(v)$ and $T_{i+}(v)$ and colors of these three set of closest vertices.

\begin{definition}
Let $d^{\inn}\in\mathbb{Z}_0^+$ and $d^{\out},d^{\sib}\in\mathbb{Z}^+$, and let three subsets of colors $C^{\inn},C^{\out},C^{\sib} \subseteq C$. A (partial) consistent subset of the subtree $T_i(v)$ with respect to the parameters $d^{\inn}$, $d^{\out}$, $d^{\sib}$, $C^{\inn}$, $C^{\out}$, $C^{\sib}$ is defined as a set of vertices $W\subseteq V(T_i(v))$ such that for any arbitrary subset $X\in \Esib(v,d^{\sib},C^{\sib})$ and $Y\in\Eout(v,d^{\out},C^{\out})$ (assuming they exist), $W$ satisfies the following (see \Cref{fig:dpp}(b)): 

\begin{mdframed}[backgroundcolor=red!10,topline=false,bottomline=false,leftline=false,rightline=false] 
\begin{itemize}
    \item $d(v,W)=d^{\inn}$. {(i.e., the distance of $v$ to its nearest member(s) in $W$ is $d^{\inn}$)}
    \item $C(\NN(v,W))=C^{\inn}$. {(i.e., $C^{\inn}$ is the set of colors of the nearest members of $v$ in the set $W$)} 
    \item For every vertex $u\in T_i(v)$, $C(u)\in C(\NN(u,W\cup X\cup Y))$. 
\end{itemize}
\end{mdframed}\label{def1}
\end{definition}

Note that for some values of $d^{\inn},d^{\out},d^{\sib},C^{\inn}, C^{\out},C^{\sib}$ there may not exist any (partial) consistent subset for $T_i(v)$; in such a case we set it as undefined. Also note that, for some values, the (partial) consistent subset can be empty as well, such as when $d^{\inn}= \infty$ and $C(u)\in C(\NN(u, X\cup Y))$ for every vertex $u\in T_i(v)$. For ease of notation, we will denote a (partial) consistent subset for $T_i(v)$ as a consistent subset with respect to the parameters $d^{\inn},d^{\out},d^{\sib},C^{\inn}, C^{\out},C^{\sib}$. 

Consider an arbitrary consistent subset $S_T$ of $T$, an arbitrary vertex $v\in V(T)$ and an integer $i\in[\eta_v]$ (see \Cref{fig:dpp}(b)).
For any vertex $v\in V(T)$ and $1\leq i\leq \eta_v$, define  
$\svin=S_T\cap V(T_i(v))$, $\svsib=S_T\cap V(T_{i+}(v))$, and $\svout=S_T\cap V(T\setminus T(v))$. Also define $\dinvs =d(v,\svin)$, $\cinvs= C(\NN(v,\svin))$, $\dsibvs=d(v,\svsib)$, $\csibvs=C(\NN(v,\svsib))$, $\doutvs =d(v,\svout)$, $\coutvs= C(\NN(v,\svout))$. Let $W$ be any arbitrary (partial) consistent subset with respect to the parameters $\dinvs,\doutvs,\dsibvs,\cinvs,\coutvs,\csibvs$ (see Definition \ref{def1}). Next, we have the following lemma.

\begin{lemma}
    $S_W=(S_T\setminus \svin)\cup W$ is a consistent subset for $T$. \label{claim:DP}
\end{lemma}

\begin{proof}
Suppose that $A= W\cup \NN(v,\svsib) \cup \NN(v,\svout)$  is the set of vertices that are either in $W$ or in the nearest neighbor of $v$ outside $T_i(v)$ in $S_W$. We will show that for any vertex $u\in T_i(v)$, $NN(u,S_W)\subseteq A$, and there is a vertex in $NN(u,A)$ of the same color as $u$. Similarly, let $B=(S_W\setminus W) \cup \NN(v,W)$. We show that for any vertex $w$ outside $T_i(v)$, $NN(w,S_W)\subseteq B$, and there is a vertex in $NN(w,B)$ of the same color as  $w$. Please note that $A$ and $B$ are not necessarily disjoint.

Consider a vertex $u\in T_i(v)$ and  $w \in T\setminus T_i(v)$. Since $\NN(v$ $,\svsib)$ $\in $ $\Esib(v,$ $\dsibvs,$ $\csibvs)$, $\NN(v,\svout) \in \Eout(v,\doutvs,\coutvs)$, and $W$ is a consistent subset with respect to the parameters $\dinvs,\doutvs,\dsibvs,\cinvs,\coutvs,\csibvs$,  we have $C(u)\in C(\NN(u,W\cup \NN(v,\svsib) \cup \NN(v,\svout)))= C(\NN(u,A))$. Also, as $S_T$ is a consistent subset, we have $C(w)\in C(\NN(w, S_T\setminus \svin)\cup \cinvs)$. From the properties of $W$, we have $C(\NN(v,W))=\cinvs$. Hence, $C(w)\in C(\NN(w, B))$. Thus, it is enough to show that (i) no vertex from $B\setminus A$ can be the closest to the vertex $u$ in the set $S_W$, and (ii) no vertex from $A\setminus B$ can be the closest vertex of $w$ in the set $S_W$.  We prove these two claims by contradiction. 

Assume that \textit{Claim (i)} is false. Then, there will be a vertex $x\in B\setminus A$, which is closest to $u$. Note that, $(B\setminus A)\cap ((T_i(v)\cup \NN(v,\svsib) \cup \NN(v,\svout))= \emptyset$. Hence we have $d(u,x)=d(u,v)+d(v,x)$, $d(v,x)>min(\dsibvs,\doutvs)$. This is a contradiction as the closest vertex from $v$ in $S_W \cup (T\setminus T_i(v))$ (if it exists) is at distance $min(\dsibvs,\doutvs) =d(v,\NN(v,\svsib) \cup \NN(v,\svout))$. Hence, $x$ cannot be the closest vertex of $u$ in $S_W$. Thus, \textit{Claim (i)} follows.

Now, assume that \textit{Claim (ii)} is false. Then there is a vertex $y \in A\setminus B$, which is closest to $w$. As $w \not\in T_i(v)$ and $y \in T_i(v)$, we have $d(w,y)=d(w,v)+d(v,y)$. Since $d(v,\NN(v,W))=\dinvs$ and $w\in (A\setminus B = W\setminus \NN(v,W))$, we have $d(v,y)>\dinvs$. This contradicts the fact that the closest vertex from $v$ in $S_W \cup T_i(v)$ (if it exists) is at distance $\din$ from $v$. Hence, \textit{Claim (ii)} is true.
\end{proof} 

Motivated by \Cref{claim:DP}, we design the following algorithm based on the dynamic programming technique. For each choice of $v\in V(T)$, $i\in[\eta_v]$, $\dinv \in [n] \cup \{0,\infty\}$, $\doutv \in [n]\cup \{\infty\}$, $\dsibv \in [n]\cup \{\infty\}$, and $\cinv,\coutv,\csibv \subseteq C$, we define a subproblem which computes the cardinality of a minimum sized (partial) consistent subset for the subtree $T_i(v)$ with respect to the parameters $\langle \dinv,\doutv,\dsibv,\cinv,\coutv,\csibv \rangle$, and 
denote its size by $P(T_{i}(v), $ $\dinv, $ $\doutv,\dsibv,\cinv,\coutv,\csibv)$.

Let us use $\dminv=\min(\dinv,\doutv,\dsibv)$.
\begin{align*}
A &= 
\begin{cases}
  \cinv   & \text{if } \dinv = \dminv\\
  \emptyset & \text{otherwise}
\end{cases},
&
B &=
\begin{cases}
  \csibv  & \text{if } \dsibv = \dminv\\
  \emptyset & \text{otherwise}
\end{cases},
&
D &=
\begin{cases}
  \coutv  & \text{if } \doutv = \dminv\\
  \emptyset & \text{otherwise}
\end{cases}
\end{align*}

We define $\cminv=A\cup B\cup D$.
{Note that $\dminv$ is the distance of the closest vertex to $v$ in any $X \in \Esib(v,\dsibv,\csibv)$, any $Y\in \Eout(v,\doutv,\coutv)$ or the consistent subset, and that $\cminv$ denotes the colors of all such vertices.} 

To compute any DP entry, we take into account the following six cases. The first two cases are for checking whether a DP entry is valid. The third case considers the scenario in which $v$ is part of the solution; the fourth, fifth, and sixth cases collectively consider the scenario in which $v$ is not in the solution.

\begin{description}
\item[Case 1:] If $C(v) \notin \cminv$, return undefined. 

\item[Case 2:] $\dinv=0$ and $\cinv\neq\{C(v)\}$. Return $\infty$.

\item[Case 3:] $\dinv=0$ and  $\cinv=\{C(v)\}$.  
Return  
\[
\begin{aligned}
P(T_{i}(v),\dinv,\doutv,\dsibv,\cinv,\coutv,\csibv)
&= 1 + \sum_{1\leq j\leq i} 
    \biggl\{
    \min_{\delta,C'} 
        P\bigl(
            T_{\eta_{v_j}}(v_j),\delta,1,\\
&\hspace{3.0cm}
            \infty, C',\{C(v)\},\emptyset
        \bigr)
    \biggr\}
\end{aligned}
\]

\end{description}
\begin{figure}[t]
                \centering
                \includegraphics[scale=1.3]{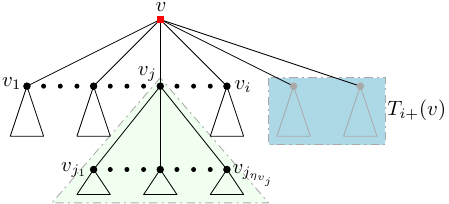}
                \caption{Illustration of the Case 3, where $\dinv=0$}
                \label{fig:dpatv}
\end{figure}

\noindent {\bf Explanation:} Case 1 and Case 2 are self-explanatory. Case 3 implies that the vertex $v$ is included in the consistent subset. Consequently, for the optimal solution, we need to determine a consistent subset for each tree rooted at a child $v_j$ of $v$, independently of each other, assuming that $v$ is part of the consistent subset (refer to \Cref{fig:dpp}(a)). For every child $v_j$ of $v$, we iterate through all possible choices of $C'\subseteq C$ and $\delta^{\inn}_{v_j}=\delta\in \{1, \ldots, h(T(v_j)\}\cup \{\infty\}$ where $h(T(v_j))$ is the height of the tree rooted at $v_j$, to identify the minimum consistent subset for $T_{\eta_{v_j}}(v_j)$. This is done with the constraints that the closest vertex in the consistent subset inside $T_{\eta_{v_j}}(v_j)$ is at a distance of $\delta$ and spans $C'$. For any vertex in $T(v_j)$, the path of the closest vertex of its own color outside $T_i(v)$ has to pass through $v$, which is considered to be in the consistent subset and has color $C(v)$. Thus, $\doutv=1$ is taken for the tree $T_{\eta_{v_j}}(v_j)$. Since we are solving for the complete tree rooted at $v_j$ with no siblings, we set the distance to the closest sibling vertex as $\dsibv=\infty$ and the corresponding color set as $\emptyset$.

\noindent {\bf Notations for Subsequent Cases:}
In the rest of the section, we consider three more cases where $\dinv>0$, and hence $\dinv,\doutv,\dsibv>0$. Intuitively, while solving the problem recursively, we will recursively solve $\mcs$ in $T_{i-1}(v)$ and $T_{\eta_{v_i}}(v_i)$. We try all possible sets of choices of $C_a, C_b$ with $\cinv=C_a\cup C_b$, and recursively solve for a solution assuming that the nodes of colors in $C_a$ are present in $T_{i-1}(v)$ at a distance of $\dinv$ (if $C_a\neq\emptyset$ and such choices are feasible) and nodes of colors in $C_b\subseteq\cinv$ are present in $T_{\eta_{v_i}}(v_i)$ at a distance of $\dinv-1$ from $v_i$ (if $C_b\neq\emptyset$ and such choices are feasible). 

\begin{figure}[h]
                \centering
                \includegraphics[scale=1.4]{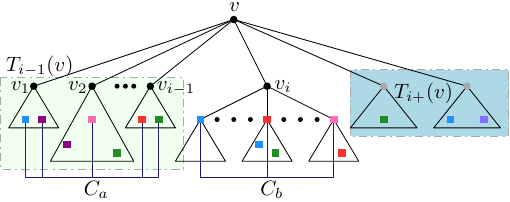}
                \caption{Illustration of the Case 4}
                \label{fig:dpnonempty}
\end{figure}
\begin{description}
\item[Case 4:] $C_a,C_b\neq \emptyset$. 
{
In this case, the closest vertices in the consistent subset from $v$ in both $T_{i-1}(v)$ and $T_{\eta({v_i})} (v_i)$ are located at a distance of precisely $\dinv$.} We start by defining the following. Let $\delta_x = \min(\dinv,\dsibv)$ and recall $\dminv=\min(\dinv,\doutv,\dsibv)$.

Observe that both $T_{i+}(v)$ and $T_{\eta_{v_i}}(v_i)$ contains siblings of $T_{i-1}(v)$. Thus, in a hypothetical consistent subset, which is compatible with the current partial consistent subset, for the tree $T$,  the closest vertices from $v$ in  $T_{(i-1)+}(v)$ are either in  $T_{i+}(v)$ or $T_{\eta_{v_i}}(v_i)$. Here, $\delta_x $ is trying to capture this distance information, and $C_{i-1}^{\sib}$ represents the colors of such vertices.
Similarly, from $v_i$ in a hypothetical consistent subset $CS$ for $T$, which is compatible with the current partial consistent subset, $\dminv+1$ denotes the distance to the vertices in $CS$ contained in $T\setminus T(v_i)$. Note that these vertices can be either in $T_{(i-1)}(v)$ or in $T_{i+}(v)$ or in $T\setminus T(v)$. $C_i^{\out}$ represents the colors of such vertices.

\begin{align*}
C_{i-1}^{\sib} &= 
  \begin{cases}
    C_b & \textrm{ if  }\dinv< \dsibv \textrm{ and  } C_b\neq \emptyset \\
    C_b \cup \csibv & \textrm{ if  }\dinv= \dsibv\\
    \csibv & \text{ otherwise }
  \end{cases}
\end{align*}

Also, define $C_i^{\out}=A\cup B\cup D$, where

\[
\begin{alignedat}{2}
A &= 
\begin{cases}
    \C_a &\text{if }\dinv = \dminv\\
    \emptyset &\text{otherwise}
\end{cases},
&\quad
B &= 
\begin{cases}
    \csibv &\text{if }\dsibv = \dminv\\
    \emptyset &\text{otherwise}
\end{cases}, \\[1ex]
D &= 
\begin{cases}
    \coutv &\text{if }\doutv = \dminv\\
    \emptyset &\text{otherwise}
\end{cases}.
\end{alignedat}
\]

Now we can safely assume that there is a  $\delta_x$-equidistant set from $v$ contained in $T_{(i-1)+}(v)$ that spans $C_{i-1}^{\sib}$. 

\[
\begin{aligned}
&\text{Return } P(T_{i}(v),\dinv,\doutv,\dsibv,\cinv,\coutv,\csibv) \\
&\quad = 
\min_{C_a,C_b} \Bigl(
    P\bigl( 
        T_{i-1}(v),\dinv,\doutv,\delta_x,C_a,\coutv,C_{i-1}^{\sib} 
    \bigr) \\
&\qquad\qquad
    + P\bigl(
        T_{\eta_{v_i}}(v_i),\dinv-1,\dminv+1,\infty,
        C_b, C_i^{\out}, \emptyset
    \bigr)
\Bigr)
\end{aligned}
\]
\end{description}

\noindent {\bf Explanation:} In this case we iterate over all possible choices of $C_a$ and $C_b$, assuming $C_a, C_b\neq \emptyset$ and $C_a\cup C_b=\cinv$. In the first part of the recursive formula, we recursively solve the problem for the tree $T_{i-1}(v)$ with the restriction that we have to include a set of vertices of color $C_a$ in $T_{i-1}(v)$ at distance $\dinv$ from $v$ (see \Cref{fig:dpp}(b)). The restriction on $\coutv$ and $\doutv$ among the vertices in $T\setminus T(v)$ remains the same as that of the parent problem. Regarding $\csibv$ and $\dsibv$,  observe that vertices in $T(v_i)$ and $T_{i+}$ are part of $T_{(i-1)+}$. Therefore the parameters for the sibling depend on the value of $\dinv$ and $\dsibv$, and accordingly, we have defined $C_{i-1}^{\sib}$.

In the second part of the recursive formula, we are solving the problem recursively for the tree $T_{\eta_{v_i}}(v_i)$, with the restriction that, in the consistent set in the consistent set we have to include a set of vertices from $T_{\eta_{v_i}}(v_i)$ which are of colors $C_b$, and at distance $\dinv-1$ from $v_i$. Observe that the vertices in $T_{i-1}(v)$, $T_{i+}(v)$ and $T\setminus T(v)$ are all outside $T(v_i)$. Thus the restriction on the distance to the vertices on the consistent subset outside $T(v_i)$ and their colors depend on the values of $\dinv,\dsibv$ and $\doutv$. Thus the distance $\doutv$ of this subproblem is defined as $\dminv=\min(\dinv,\dsibv,\doutv)$, and the set of colors $C_i^{\out}$ is defined accordingly. As we are solving for the whole tree rooted at $v_i$, there are no siblings; so $\dsibv=0$ and $C_i^{\sib}=\emptyset$.

\begin{description}
\item[Case 5:] $C_a=\emptyset$ and $C_b=\cinv$.
{Note that in this case, the closest vertices in the consistent subset from $v$ in $T_{\eta(v_i)}(v_i)$ are located at a distance of $\dinv$ while in $T_{i-1}(v)$, they are located at a distance of at least $\delta \geq \dinv+1$}

We iterate over all values $\delta>\dinv$ and all possible choices of colors to find the size of a minimum consistent subset. 
We define $\delta_x$ and $C_{i-1}^{\sib}$ the same as in Case 4. 
For any values $\delta>\dinv$ and $C\subseteq [c]$, we define $\dminv(\delta,C)=\min(\delta,\dsibv,\doutv)$, and $C_i^{\out}(\delta,C)=A\cup B\cup D$ where

\[
\begin{alignedat}{2}
A &= 
\begin{cases}
    \C & \text{if }\delta = \dminv\\
    \emptyset & \text{otherwise}
\end{cases},
&\quad
B &= 
\begin{cases}
    \csibv & \text{if }\dsibv = \dminv\\
    \emptyset & \text{otherwise}
\end{cases}, \\[1ex]
D &= 
\begin{cases}
    \coutv & \text{if }\doutv = \dminv\\
    \emptyset & \text{otherwise}
\end{cases}.
\end{alignedat}
\]

From $v_i$ in a hypothetical consistent subset $CS$ for $T$, which is compatible with the current partial consistent subset, $\dminv(\delta,C)$ denotes the distance to the vertices in $CS$ contained in $T\setminus T(v_i)$. Note that these vertices can be either in $T_{(i-1)}(v)$ or in $T_{i+}(v)$ or in $T\setminus T(v)$. $C_i^{\out}(\delta,C)$ represents the colors of such vertices.

\[
\begin{aligned}
& P(T_{i}(v),\dinv,\doutv,\dsibv,\cinv,\coutv,\csibv) \\
&\quad = 
\min_{\delta>\dinv,\, C\subseteq [c]}
\Bigl(
    P\bigl( 
        T_{i-1}(v),\delta,\doutv,\delta_x,C,\coutv,C_{i-1}^{\sib} 
    \bigr) \\
&\qquad\qquad
    + P\bigl(
        T_{\eta_{v_i}}(v_i),\dinv-1,\dminv(\delta,C)+1,\infty,
        C_b, C_i^{\out}(\delta,C), \emptyset
    \bigr)
\Bigr)
\end{aligned}
\]
\end{description}
\noindent {\bf Explanation:} The explanation for this case is the same as Case 4 except for the fact that we have to make sure that the closest vertex chosen in the consistent subset from $T_{i-1}(v)$ is at distance at least $\dinv+1$.

{}

\begin{description}
\item[Case 6:] $C_b=\emptyset$ and $C_a=\cinv$.

Here we consider the case when $C_b=\emptyset$ and $C_a=\cinv$. 
Note that in this case, the closest vertices in the consistent subset from $v$ in $T_{i-1}(v)$ are located at a distance of $\dinv$ while in $T_{\eta(v_i)}(v_i)$, they are located at a distance of at least $\delta \geq \dinv+1$

We define
$\dminv(\delta,C)=\min(\dinv,\dsibv,\doutv)$. We define $C_i^{\out}(\delta,C)=A\cup B\cup D$ where

\begin{align*}
\begin{alignedat}{2}
A &= 
\begin{cases}
    \C_a &\!\text{if }\dinv = \dminv\\
    \emptyset &\!\text{otherwise}
\end{cases},
&\quad
B &= 
\begin{cases}
    \csibv &\!\text{if }\dsibv = \dminv\\
    \emptyset &\!\text{otherwise}
\end{cases}, \\[1ex]
D &= 
\begin{cases}
    \coutv &\!\text{if }\doutv = \dminv\\
    \emptyset &\!\text{otherwise}
\end{cases}.
\end{alignedat}
\end{align*}

For any values $\delta>\dinv$ and $C\subseteq [c]$ we define
$\delta_x(\delta,C)=\min(\dsibv,\delta)$ We define $C_{i-1}^{\sib}(\delta,C)=E\cup F$ where
\begin{align*}
  E = 
  \begin{cases}
    \C  \textrm{ if  }\delta= \delta_x(\delta,C)\\
    \emptyset \textrm{ otherwise  }
  \end{cases}, \text{ } \text{ } \text{ }
  &F = 
  \begin{cases}
    \csibv \textrm{ if  }\dsibv= \delta_x(\delta,C)\\
    \emptyset \text{ } \text{ } \text{ } \textrm{ otherwise  }
  \end{cases}
\end{align*}

The meaning and the reasoning behind the defining of $\dminv(\delta,C)$ and $C_i^{\out}(\delta,C)$ remains the same as in previous cases.  
Thus, in a hypothetical consistent subset, which is compatible with the current partial consistent subset, for the tree $T$,  the closest vertices from $v$ in  $T_{(i-1)+}(v)$ are either in  $T_{i+}(v)$ or $T_{\eta_{v_i}}(v_i)$. Here, $\delta_x(\delta,C) $ is trying to capture this distance information, and $C_{i-1}^{\sib}(\delta,C)$ represents the colors of such vertices.

In this case we return:\\ 

Return $P(T_{i}(v),\dinv,\doutv,\dsibv,\cinv,\coutv,\csibv)=$ 

\hspace*{0pt}\hfill$\>\displaystyle{
\min_{\delta>\dinv, \text{ }C\subseteq [c]}\Bigl(P\Bigl( 
        T_{i-1}(v),\dinv,\doutv,\delta_x(\delta,C),C_a ,\coutv,C_{i-1}^{\sib}(\delta,C) 
        \Bigr)} $
        
        \hspace*{\fill}\>$\displaystyle{+P\Bigl(
        T_{\eta_{v_i}}(v_i),\delta,\dminv+1,\infty,C_b, C, \emptyset
        \Bigr)\Bigl)}$

\end{description}
\noindent {\bf Explanation:} The explanation for this case is the same as the previous two cases.

\noindent{\bf Running time of the algorithm} 
The total number of choices of $\dinv, $ $ \doutv, $ $ \dsibv, \cinv, \coutv$ and $\csibv$ is bounded by $n^3 2^{3c}$. For each choice $C_a$ and $C_b$, the algorithm takes at-most $n2^c$ time to go through all possible entries of $\delta$ and $C$ (in \textbf{case 5} and \textbf{case 6}) and there are at-most $2^{2c}$ choices of $C_a$ and $C_b$. The recursion runs for at most $n^2$ times. Hence, the worst-case running time of the algorithm is $\mathcal{O}(2^{6c}n^6)$. Hence, by combining all the above six cases, we obtain the following theorem.

\begin{theorem}
Given an unweighted, undirected tree $T = (V(T),E(T))$ with $n$ vertices and $c$ color classes, there exists a fixed-parameter tractable (FPT) algorithm that computes a Minimum Consistent Subset (\mcs) in time $\mathcal{O}(2^{6c}n^6)$ when $c$ is the parameter.
\end{theorem}

\section{\np-hardness of \mcs on Interval Graphs} \label{hardnessint}
A graph $H$ is said to be an interval graph if there exists an interval layout of the graph $H$, or in other words, for each node, $v_i \in V(H)$ one can assign an interval $\alpha_i$ on the real line such that $(v_i,v_j) \in E(H)$ if and only if $\alpha_i$ and $\alpha_j$ (completely or partially) overlap in the layout of those intervals.   

We prove that the Minimum Consistent Subset problem is \npc even when the input graph is an interval graph. We present a reduction from the Vertex Cover problem for cubic graphs. It is known that Vertex Cover remains NP-complete even for cubic graphs~\cite{GAREY1976237}. For any set of intervals $\mathcal{I}$, let $G(\mathcal{I})$ be the interval graph corresponding to the set of intervals $\mathcal{I}$.

\begin{tcolorbox}[breakable,bicolor,
  colback=cyan!5,colframe=cyan!5,title=Interval Graph Construction.,boxrule=0pt,frame hidden]

\noindent \underline{\textbf{Interval Graph Construction.}}

\noindent Let $G$ be any cubic graph, where $V(G)=\{v_1, \ldots, v_n\}$ is the set of vertices, and $E(G)=\{e_1, \ldots, e_m\}$ is the set of edges in $G$. We create the set of intervals $\mI_G$ for $G$ on a real line $\cal L$. The set of intervals in $\mI_G$ is represented by intervals of three different sizes, {\em medium}, {\em small} and {\em large}, where each medium interval is of unit length, each small interval is of length $\epsilon<<\frac{1}{2n^3}$ and the length of the large interval is $\ell>>2n$.

We define $\mI_G=\mI_1\cup \mI_2\cup \mI_3\cup \mI_4$ where $\mI_1$ contains $2m$ {\em medium} size intervals (two intervals for each edge) and defined as  $\mI_1=\{I(e_i,v_j):e_i=(v_j,y)\in E(G), y\in V(G)\}$.  We set color $c_i$ to the interval $I(e_i,v_j)$. $\mI_2$ contain $n\cdot n^3$ {\em small} intervals of color $c_{m+1}$, and $\mI_3$ contain $n\cdot n^4$ {\em small} intervals of color $c_{m+1}$. $\mI_4$ contains one large interval $I_\ell$ of color $c_1$.

We create the following vertex gadget $X_i$ for each vertex $v_i\in V(G)$. $X_i$ contains the following {\em medium} size intervals $\{I(e,v_i):e=(v_i,x)\in E(G)\}$ corresponding to the edges that are incident on $v_i$. These intervals span the same region $s_i$ of unit length on the real line $\cal L$; hence, they are mutually completely overlapping. In the vertex gadget $X_i$, we also include a total of $n^3$ mutually non-overlapping small intervals in the set $\mI_2$. Span of all the small intervals in $X_i$ is contained in the span $s_i$ of the {\em medium} sized intervals in $X_i$ (see \Cref{fig:1K}).

Each vertex gadget is placed one after another (in a non-overlapping manner) along the line $\cal L$ in an arbitrary order, such that a total of $n^4$ mutually non-overlapping small intervals can be drawn between two consecutive vertex gadgets. Thus, $\mI_3$ contains $n$ sets of $n^4$ non-overlapping small intervals. Finally, $\mI_4$ contains a single large interval 
$I_\ell$ that contains all the intervals in $\mI_1 \cup \mI_2\cup \mI_3$. This completes the construction.

\end{tcolorbox}
 
\begin{figure}[!h]
\centering
\includegraphics[width=1\textwidth]{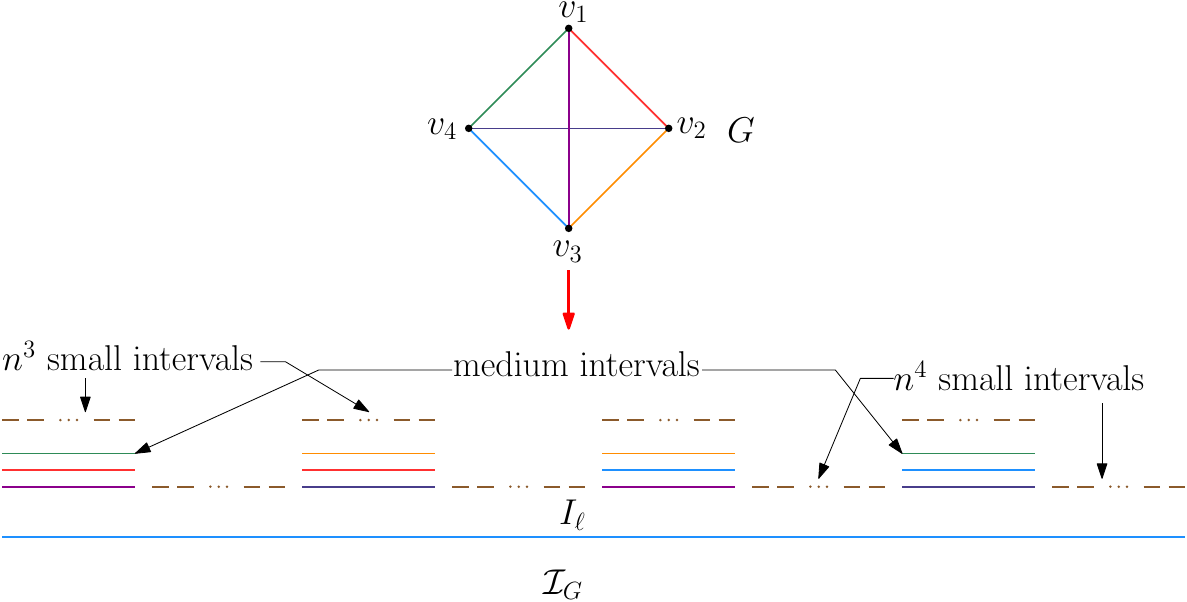}
\caption{An example reduction}
\label{fig:1K}
\end{figure}

\begin{lemma}\label{lemma:intervallemma}
The graph $G$ has a vertex cover of size at most $k$ if and only if the corresponding interval graph $G(\mI_G)$ has a consistent subset of size at most $K=k(3+n^3)$.
\end{lemma}
\begin{proof}
 With a slight abuse of notations, we will denote the vertex in $G(\mI_G)$ corresponding to an interval $I\in \mI_G$ by $I$.
$(\Rightarrow)$
Let $A \subseteq V(G)$ be a vertex cover of $G$. Consider the set of intervals $\mI_A=\bigcup_{v_i \in A} X_i$. We prove that $\mathcal{I}_A$ is a consistent subset of $G(\mI_G)$. Note that as $|X_i|=3+n^3$, $|\mathcal{I}_A|=k(3+n^3)$.

As the vertices in $A$ cover all the edges in $G$, $\mI_A$ must contain at least one interval of each color $\{c_1, \cdots, c_m\}$. As the unit intervals in $X_i$ associated with a vertex $v_i\in A$ are of colors different from the color of the small intervals in $X_i$, $\mI_A$ contains at least $n^3$ small intervals. Therefore $\mI_A$ contains at least one interval from each color in $\{c_1, \ldots c_{m+1}\}$.
Observe that, (i) the interval $I_\ell\in \mI_4$ of color $c_1$ contains an interval of color $c_1$ that corresponds to the edge $e_1$, and (ii) the distance between any two nodes  
corresponding to medium intervals in two different vertex gadgets of $G(\mI_G)$ is $2$
(via the node corresponding to the interval $I_\ell$ in $G(\mI_G)$). Thus, $\mI_A$ is a consistent subset.

    \smallskip
\noindent $ (\Leftarrow)$ 
Let $\mI_B\subseteq \mI_G$ be any consistent subset of $G(\mI_G)$ of cardinality $(3+n^3)k$. 
Now, if $\mI_B$ contains $I_\ell$ then $\lvert \mI_G\rvert-2\leq \lvert \mI_B\rvert\leq\lvert \mI_G\rvert$ because if $e_1=(v_i,v_j)$ be the edge of color $c_1$, then we can only do not take the medium intervals of color $c_1$ from the vertex gadget $X_i$ and $X_j$ in $\mI_B$ because they are covered by $I_\ell$, which contradicts the fact that $|\mI_B|=(3+n^3)k$.
Thus, we have $I_\ell\notin \mI_B$.

By the definition, $\mI_B$ contains at least one color from $\{c_1,\cdots, c_m,c_{m+1}\}$. Also, if $\mI_B$ contains one interval from the gadget $X_i$ of any vertex $v_i$, then it must contain all the intervals from $X_i$; otherwise, it can not be a consistent subset. Thus $\mI_B$ is the union $\mathcal{X}$ of at most $k$ sets from $\{X_i:i\in[n]\}$, and it contains at least one interval from each color $\{c_1,\cdots,c_m\}$, and a few intervals of color $c_{m+1}$. Now, consider the set $V_B=\{v_i:X_i\in \mathcal{X}\}$, which is a vertex cover for the graph $G$ of size at most $k$. 

Hence, the lemma is proved.
\end{proof}

\section{log-APX Approximation of \mscs on general graphs}\label{section7}

We reduce an instance of the \emph{Set Cover} problem to an instance of the Strict Consistent Subset problem of a graph $G=(V(G), E(G))$ in polynomial time. The graph $G$ has two \emph{red} vertices, while all others are \emph{blue}. It is known that the Set Cover problem is \lapxh, i.e., it is $\nph$ to approximate within a factor of $c \cdot \log n$ \cite{raz97}.
    
\textbf{Reduction.}
 Given a set cover instance with $n$ elements $I=\{e_1, e_2, \dots, e_n\}$ and $m$ sets $S=\{S_1, S_2, \dots, S_m\}$, where each $S_i$ is a subset of $I$ and their union is $I$. We construct a graph $G=(V(G),E(G))$ as follows (see \Cref{fig:1G}(a)):
\begin{enumerate}
        \item Initially $V(G) :=\emptyset$ and $E(G) :=\emptyset$.
        
        \item For each $e_i\in I$, create a vertex $x_i$ in a set $V_1$, resulting in $n$ vertices.

        \item For each set $S_j\in S$, create a vertex $y_j$ in a set $V_2$, resulting a total of $m$ vertices.

        \item Add an edge between $y_j$ and $x_i$ if $e_i\in S_j$.
        
        \item Connect every pair of vertices in $V_2$, forming a clique.
        
        \item Add a \emph{red} vertex $r_1$ in a set $V_3$ and connect it to all $n$ vertices in $V_1$.
        
        \item Add another \emph{red} vertex $r_2\in V_3$ and connect it to $r_1$.
        \item Set $V(G)=V_1\cup V_2\cup V_3$, resulting a total of $n+m+2$ vertices.

        \item Assign \emph{blue} color to all vertices in $V_1$ and $V_2$.

    \end{enumerate}
\begin{figure}
\centering
\includegraphics[width=.4\textwidth]{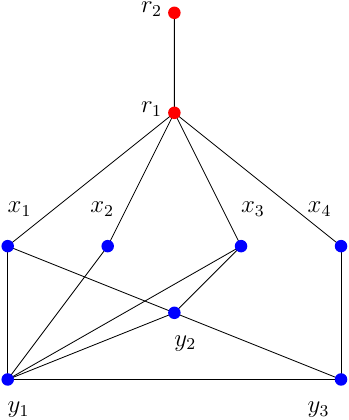}
\caption{Reduction from an instance of the Set Cover problem to an instance of the Strict Consistent Subset problem with $I=\{e_1,e_2,e_3,e_4\}$ and $S_1=\{e_1,e_2, e_3\}$, $S_2=\{e_1,e_3\}$, $S_3=\{e_4\}$.}\label{fig:1G}
\end{figure}
    
This completes the construction of $G=(V(G), E(G))$ from the set cover instance.

\begin{lemma}{\label{mannalemma6}}
The set system $\{I,S\}$ has a set cover of size $k$ if and only if the graph $G=(V(G),E(G))$ has a strict consistent subset of size $k+1$
\end{lemma}

\begin{proof}
If there exists a Set Cover $X$ of size $k$, we construct a strict consistent subset $Y$ of size $k+1$ by including the vertices of $V_2$ corresponding to the sets in the cover, together with $r_2$. Since all vertices of $V_2$ are at distance two from $r_1$, the nearest vertex to $r_1$ in $Y$ is $r_2$. Every vertex of $V_1$ has a nearest neighbor in $Y$ since the vertices of $V_2$ chosen in $Y$ correspond to a set cover of ${I,S}$, and $V_2$ being a clique ensures all unselected vertices of $V_2$ have a neighbor in $Y$, making $Y$ a strict consistent subset.

Conversely, suppose $Y$ is a strict consistent subset of size $k+1$. If $Y$ contains $r_1$, it must include all the vertices of $V_1$, yielding a strict consistent subset of size at least $n+1$. So, assume $r_1\notin Y$, implying $r_2 \in Y$ according to \Cref{bubailemma101}. Since $r_2\in Y$, no vertex of $V_1$ is in $Y$; otherwise, $r_1$ would have a nearest neighbor in $Y$ of a different color. Thus $Y$ contains no vertices of $V_1$, each vertex of $V_1$ must have a neighbor in $V_2$ within $Y$. The sets corresponding to the vertices of $V_2$ in $Y$ form a set cover of $\{I,S\}$ of size $k$.
\end{proof}

\begin{theorem}
The \mscs problem is $\lapxh$.
\end{theorem}

\begin{proof}
Since  the Set Cover problem is \lapxh \cite{raz97}, Lemma~\ref{mannalemma6} implies that the \mscs problem is also \lapxh.
\end{proof}

\section{NP-hardness of MSCS on Planar Graphs}\label{npmscs}

A planar graph is one that can be drawn in the plane without edge crossings. We prove that the Minimum Strict Consistent Subset (\mscs) problem is $\npc$ on planar graphs by a reduction from the \textsc{Dominating Set} problem on planar graphs, which is known to be NP-complete~\cite{article}.

Let $G=(V(G),E(G))$ be a planar graph with $n = |V(G)|$, and an integer $k > 0$. The \textsc{Dominating Set} problem asks whether $G$ has a dominating set of size at most $k$. We construct a corresponding vertex-colored graph $G'$ from $G$ as follows.

\begin{figure}
\centering
\includegraphics[width=.75\textwidth]{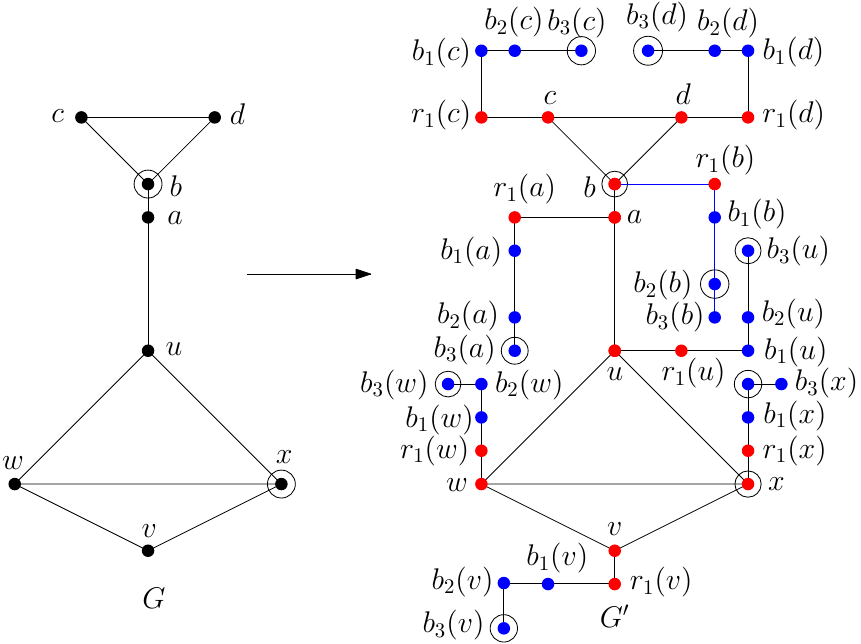}
\caption{Reduction from an instance of the Dominating Set problem on a planar graph $G$ to an instance of the Strict Consistent Subset problem on a planar graph $G'$. Vertices inside the circle represent those chosen in the solution.}
\label{fig:mscshphard}
\end{figure}

\begin{tcolorbox}[breakable,bicolor,
  colback=cyan!5,colframe=cyan!5,title=Interval Graph Construction.,boxrule=0pt,frame hidden]

\noindent \underline{\textbf{Planar Graph Construction.}}

\vspace{0.3em}

For each vertex $v \in V(G)$, we add a path of four new vertices to $G'$:
\[
r_1(v) \text{ -- } b_1(v) \text{ -- } b_2(v) \text{ -- } b_3(v),
\]
and connect the original vertex $v$ to $r_1(v)$ with an edge. The coloring of the vertices in $G'$ is assigned as follows (see Figure \ref{fig:mscshphard}):
\begin{itemize}
    \item Every original vertex $v \in V(G)$ and each new vertex $r_1(v)$ are colored \textbf{red}.
    \item Each new vertex $b_1(v)$, $b_2(v)$, and $b_3(v)$ is colored \textbf{blue}.
\end{itemize}

\end{tcolorbox}

We now establish essential lemmas and the main theorem for the reduced graph $G'$.

\begin{lemma}\label{lemmamscs1}
    For any strict consistent subset $S$ of $G'$ and for any vertex $v \in V(G)$, either both $r_1(v)$ and $b_1(v)$ are in $S$, or neither is in $S$.
\end{lemma}
\begin{proof}
    Assume, for the sake of contradiction, that exactly one of $b_1(v)$ or $r_1(v)$ is in $S$ for some vertex $v$. First, suppose $b_1(v) \in S$ but $r_1(v) \notin S$. Since $r_1(v)$ is adjacent to $b_1(v)$ and is red, one of its nearest neighbors in $S$ would be $b_1(v)$, which is a blue vertex of a different color. This violates the condition for $S$ to be a strict consistent subset. Similarly, if $r_1(v) \in S$ but $b_1(v) \notin S$, then the blue vertex $b_1(v)$ would have $r_1(v)$ as one of its nearest neighbors in $S$, which is also forbidden since their colors differ. Therefore, both must be included together or excluded together.
\end{proof}

We now present the main lemma of the reduction.

\begin{lemma}
    The graph $G$ has a dominating set of size $k$ if and only if the colored graph $G'$ has a strict consistent subset of size $n + k$.
\end{lemma}
\begin{proof}
    ($\Rightarrow$) Suppose $G$ has a dominating set $D \subseteq V(G)$ with $|D| = k$. We construct a strict consistent subset $S$ for $G'$ as follows:
    \begin{itemize}
        \item For each vertex $v \in D$: add $v$ (red) and $b_2(v)$ (blue) to $S$.
        \item For each vertex $v \notin D$: since $D$ is a dominating set, there exists a vertex $u \in D$ such that $u$ and $v$ are adjacent in $G$ ($uv \in E(G)$). For such a $v$, add $u$ (red, but note $u$ is already added if $u \in D$) and $b_3(v)$ (blue) to $S$.
    \end{itemize}
    Hence $|S| = n + k$, consisting of $n$ blue vertices ($b_2(v)$ for $v \in D$ and $b_3(v)$ for $v \notin D$) plus the $k$ red vertices in $D$. We now verify that $S$ is a strict consistent subset of $G'$.
    \begin{itemize}
        \item For a vertex $v \in D$:
        \begin{itemize}
            \item The red vertex $r_1(v)$ has only $v$ as its nearest neighbor in $S$.
            \item The blue vertex $b_1(v)$ has only $b_2(v)$ as its nearest neighbors in $S$.
            \item The blue vertex $b_2(v)$ is in $S$.
            \item The blue vertex $b_3(v)$ has only $b_2(v)$ as its nearest neighbors (as well as only adjacent vertex) in $S$.
        \end{itemize}
        \item For a vertex $v \notin D$ with dominating neighbor $u \in D$:
        \begin{itemize}
            \item The red vertex $v$ has $u$ (also red) as its nearest neighbor (and adjacent vertex) in $S$, and no blue vertices in $G'$ are adjacent to $v$.
            \item The red vertex $r_1(v)$ has $u$ as its nearest neighbor vertex ($d(r_1(v),$ $ u) $ $ \leq 2$ via the vertex $v$, while the distance to any blue vertex in $S$ is at least 3) in $S$.
            \item The blue vertices $b_1(v)$ and $b_2(v)$ have $b_3(v)$ as their nearest neighbor vertices in $S$.
            \item The blue vertex $b_3(v)$ is in $S$.
        \end{itemize}
    \end{itemize}
    Thus, every vertex $v$ in $G'$ has a nearest neighbor in $S$ of the same color, so $S$ is a valid strict consistent subset.

    \medskip
    ($\Leftarrow$) Conversely, suppose $G'$ has a strict consistent subset $S$ with $|S| = n + k$. By Lemma~\ref{bubailemma101}, for each $v \in V(G)$, the set $S$ must contain at least one vertex from the block $B(v) = \{b_1(v), b_2(v), b_3(v)\}$. Let $B = V(G) \cup \{r_1(v): v \in V(G)\}$ be the block of red vertices. Thus, $G'$ has $n+1$ such blocks. We consider three cases:
    \begin{enumerate}
        \item If $b_1(v) \in S$, then Lemma~\ref{lemmamscs1} implies that $r_1(v) \in S$ as well.
        \item If $b_2(v)$ is the only vertex from $B(v)$ in $S$, then $v$ must be in $S$; otherwise, $r_1(v)$ or $v$ would have its nearest neighbor in $S$ of a different color.
        \item If $b_3(v)$ is the only vertex from $B(v)$ in $S$, then at least one neighbor of $v$ in $V(G)$ must also be in $S$; otherwise, $r_1(v)$ would have $b_3(v)$ as its nearest neighbor.
    \end{enumerate}
    If $r_1(v)$ and $b_1(v)$ are both in $S$ for all $v \in V(G)$, then $|S| \geq 2n$, contradicting $|S| = n+k$ for $n > k$, which holds whenever $G$ has at least one edge. Therefore, for some vertices $v$, at least one of $b_2(v)$ or $b_3(v)$ (instead of $b_1(v)$) must belong to $S$. Thus, at most $k$ additional vertices in $S$ are red. 
    
    Let $D' = S \cap (V(G) \cup \{r_1(v): v \in V(G)\})$ be the set of red vertices in $S$. Then $|D'| \leq k$. We now construct a set $D$ from $D'$, which will be a dominating set of $G$ of size at most $k$. Initialize $D := \emptyset$ and proceed as follows:
    \begin{itemize}
    \item If $b_1(v)\in S$, then $r_1(v)\in S$. In such a case, include $v$ in $D$.
    \item If $b_1(v) \notin S$ but $b_2(v)\in S$, then $v\in S$. In such a case, include $v$ in $D$.
    \item If $b_1(v)\notin S$ and $b_2(v)\notin S$, but $b_3(v)\in S$, then at least one adjacent vertex of $v$ from $V(G)$ must be in $S$. In such a case, include one adjacent vertex of $v$ from $V(G)$ in $D$.
\end{itemize}
    Thus $|D|\leq k$. We claim that $D$ is a dominating set of $G$, because for every $v \in V(G)$, either $v \in D$ or at least one neighbor of $v$ is in $D$. Since $S$ already contains at least $n$ blue vertices and $|D|\leq k$, we may add extra vertices from $V(G)$, if necessary, to obtain $|D| = k$.
\end{proof}
\begin{theorem}
    The \mscs problem is $\npc$ on bichromatic planar graphs.
\end{theorem}
    
\begin{proof}
    It is easy to see that the problem is in $\np$. Since the \mscs problem is $\nph$ by the above reduction, it follows that the problem is $\npc$.
\end{proof}

\section{Conclusion} 
We have shown that \mcs is NP-complete for trees and interval graphs, and have given an exact algorithm parameterized with respect to the number of colors for trees.
As a direction for future research, possibilities for approximation algorithms for \mcs can be explored on interval graphs and related graph classes such as circle graphs and circular-arc graphs. Parameterized algorithms may also be explored where, in addition to a parameter for the number of colours, there is also a parameter specifying the structural properties of the input graph.

Also, since \mscs is $\npc$ on planar graphs, designing approximation algorithms for planar graphs and hardness results on other graph classes appear to be promising directions. However, from the $\npc$ reduction, we observe that \mscs remains $\npc$ even when $c = 2$. Thus, no \fpt algorithm exists when the number of colors is taken as the parameter on planar graphs. Nevertheless, FPT results may still be achievable with respect to other parameters.

\bibliographystyle{elsarticle-num} 
\bibliography{ref}






\end{document}